\documentclass{article}
\usepackage{mlspconf}
\usepackage{amsfonts}
\usepackage{amsmath}
\usepackage{graphicx}
\usepackage{algpseudocode}
\usepackage{color}
\usepackage{enumitem}
\usepackage[makeroom]{cancel}
\usepackage{MnSymbol}

\usepackage{dblfloatfix}

\usepackage{amsthm}

\usepackage{url}

\newtheoremstyle{break}
  {}
  {}
  {\itshape}
  {}
  {\bfseries}
  {.}
  {\newline}
  {}
\theoremstyle{break}
\newtheorem{proposition}{Proposition}

\theoremstyle{definition}
\newtheorem{definition}{Definition}

\setlist[itemize]{label=$\triangleright$}

\DeclareMathOperator*{\argmin}{arg\min}
\DeclareMathOperator*{\argmax}{arg\max}
\newcommand{\bs}{\boldsymbol}
\newcommand{\E}{\mathbb{E}}

\title{A variance modeling framework based on variational autoencoders for speech enhancement}
\name{Simon Leglaive\textsuperscript{\normalfont 1}\thanks{This work is supported by the ERC Advanced Grant VHIA \#340113.} \qquad Laurent Girin\textsuperscript{\normalfont 1,2} \qquad Radu Horaud\textsuperscript{\normalfont 1}}

\address{\textsuperscript{1}Inria Grenoble Rh\^one-Alpes, France \qquad \textsuperscript{2}Univ. Grenoble Alpes, Grenoble INP, GIPSA-lab, France}
\begin{document}
\ninept

\maketitle
\begin{abstract}
In this paper we address the problem of enhancing speech signals in noisy mixtures using a source separation approach. We explore the use of neural networks as an alternative to a popular speech variance model based on supervised non-negative matrix factorization (NMF). More precisely, we use a variational autoencoder as a speaker-independent supervised generative speech model, highlighting the conceptual similarities that this approach shares with its NMF-based counterpart. In order to be free of generalization issues regarding the noisy recording environments, we follow the approach of having a supervised model only for the target speech signal, the noise model being based on unsupervised NMF. We develop a Monte Carlo expectation-maximization algorithm for inferring the latent variables in the variational autoencoder and estimating the unsupervised model parameters. Experiments show that the proposed method outperforms a semi-supervised NMF baseline and a state-of-the-art fully supervised deep learning approach.

\end{abstract}
\section{Introduction}
\label{sec:intro}

Speech enhancement is a classical problem of speech processing, which aims to recover a clean speech signal from the recording of a noisy signal, where the noise is generally considered as additive \cite{loizou2007speech}. In this work we address single-channel speech enhancement, which can be seen as a an under-determined source separation problem, where the sources to be separated are of different nature. 

Statistical approaches combining a local Gaussian modeling of the time-frequency signal coefficients with a variance model have attracted a lot of attention in the past few years \cite{vincent2010probabilistic}. Within this framework, non-negative matrix factorization (NMF) techniques are especially popular to structure the time-frequency-dependent signal variance \cite{ISNMF}. Recently, deep neural networks were also successfully investigated for speech enhancement, to generate either time-frequency masks or clean power spectrograms from noisy spectrograms \cite{wang2017supervised}, or to model the signal variance \cite{nugraha_taslp16} (here in a multichannel framework). Similarities between NMF and standard autoencoders have been reported in \cite{Smaragdis:alternative}. The recent approach \cite{bando2017statistical} also falls into the variance modeling framework. Very interestingly, the authors developed a semi-supervised single-channel speech enhancement method mixing concepts of Bayesian inference and deep learning, through the use of a variational autoencoder \cite{kingma2013auto}. 

In the present work, we follow a similar approach, using a speaker-independent supervised Gaussian speech model based on a variational autoencoder. For a given time-frequency point, it mainly consists in modeling the speech variance as a non-linear function of a Gaussian latent random vector, by means of a neural network. Compared with \cite{bando2017statistical}, we further introduce a time-frame-dependent gain in order to provide some robustness with respect to the loudness of the training examples. The noise model is Gaussian and based on unsupervised NMF, so as to be free of generalization issues regarding the noisy recording environments. We propose a new Monte Carlo expectation-maximization algorithm \cite{wei1990monte} for performing semi-supervised speech enhancement. Experimental results show the superiority of our approach compared with a baseline method using NMF, and also compared with a state-of-the-art fully supervised deep learning approach \cite{xu2015regression}.

We first provide in Section~\ref{sec:varianceModelingFramework} a brief presentation of the NMF-based variance modeling framework commonly used in audio source separation. This will help us emphasizing the conceptual similarities shared with the speech model used in this work. In Section~\ref{sec:Model} we present the model. Our inference algorithm is detailed in Section~\ref{sec:inference}. The optimization problem related to the variational autoencoder training is presented in Section~\ref{sec:training}. Experiments are described in Section~\ref{sec:experiments} and we finally conclude in Section~\ref{sec:conclusion}.

\section{Variance Modeling Framework}
\label{sec:varianceModelingFramework}

We work in the short-term Fourier transform (STFT) domain. For all $(f,n) \in \mathbb{B}={\{0,...,F-1\}}\times{\{0,...,N-1\}}$, where $f$ denotes the frequency index and $n$ the time-frame index, the single-channel speech enhancement problem consists in recovering the clean speech STFT coefficients $s_{fn} \in \mathbb{C}$ from the noisy observed ones $x_{fn} = s_{fn} + b_{fn}$, where $b_{fn} \in \mathbb{C}$ is a noise signal. The source STFT coefficients are modeled as complex circularly symmetric Gaussian random variables \cite{vincent2010probabilistic}:
\begin{equation}
s_{fn} \sim \mathcal{N}_c(0, v_{s,fn}); \qquad
b_{fn} \sim \mathcal{N}_c(0, v_{b,fn}),
\label{local_gaussian_source_model}
\end{equation}
where $\mathcal{N}_c(\mu, \sigma^2)$ denotes the complex proper Gaussian distribution, which is circularly symmetric (rotational invariant in the complex plane) if $\mu = 0$. Under a local stationarity assumption \cite{liutkus2011gaussian}, the variances in \eqref{local_gaussian_source_model} characterize the short-term power spectral density of the audio source signals. 
Note that this Gaussian model has been recently extended to other complex circularly symmetric and elliptically contoured distributions \cite{ollila2011complex}, such as the symmetric alpha-stable \cite{liutkus2015generalized} or the Student's \textit{t} \cite{yoshii2016student} distributions. 

A widely used linear variance model for the source signals relies on NMF. For $j\in\{s,b\}$, let $\mathbf{W}_{j} \in \mathbb{R}_+^{F \times K_j}$ be a matrix containing $K_j$ spectral patterns, and let $\mathbf{H}_{j} \in \mathbb{R}_+^{K_j \times N}$ be a matrix containing the activations of these spectral patterns over the time frames. At each time-frequency point $(f,n) \in \mathbb{B}$, the variance $v_{j,fn}$ is modeled as follows: 
\begin{equation}
v_{j,fn} =  \left(\mathbf{W}_{j} \mathbf{H}_{j}\right)_{f,n} = \mathbf{w}_{j,f}^\top \mathbf{h}_{j,n},
\label{NMF_variance_model}
\end{equation}
where $\mathbf{w}^\top_{j,f} \in \mathbb{R}_+^{K_j}$ is $f$-th row of $\mathbf{W}_{j}$, $\mathbf{h}_{j, n} \in \mathbb{R}_+^{K_j}$ is $n$-th column of $\mathbf{H}_{j}$, and $\cdot^\top$ denotes the transposition operator. Without any additional constraint on the model parameters, the NMF rank $K_j$ is generally chosen such that $K_j(F + N) \ll FN$. Compared with the unconstrained Gaussian model in \eqref{local_gaussian_source_model}, the number of parameters to be estimated is therefore considerably reduced, which is of great interest for solving an under-determined problem.

In a semi-supervised speech enhancement setting \cite{smaragdis2007supervised}, $\mathbf{W}_{s}$ is learned on a clean speech database. The variance model in \eqref{NMF_variance_model} is therefore seen as a linear function of the activations. Then, in the speech enhancement algorithm, only the speech activation matrix $\mathbf{H}_{s}$ and the noise NMF parameters $\mathbf{W}_{b}$ and $\mathbf{H}_{b}$ are estimated from the observation of the noisy mixture signal. It has been shown in \cite{ISNMF} that maximum likelihood estimation of the NMF parameters under this Gaussian model amounts to solving an optimization problem involving the Itakura-Saito divergence. An optimal estimate of the speech signal in the minimum mean square error sense is then obtained by Wiener filtering.

\section{Model}
\label{sec:Model}

\subsection{Supervised Speech Model}

In a supervised setting, the variance model in \eqref{NMF_variance_model} is limited by its linear nature. Moreover, the number of trainable parameter is restricted by the factorization rank, which is usually chosen to be small. In this work we consider a supervised and non-linear modeling of the speech variances. Independently for all $(f,n) \in \mathbb{B}$, we have the following generative model, involving a latent random vector $\mathbf{z}_n \in \mathbb{R}^L$:
\begin{align}
\mathbf{z}_n &\sim \mathcal{N}(\mathbf{0}, \mathbf{I}); \label{prior_VAE} \\
s_{fn} \mid \mathbf{z}_n ; \bs{\theta}_s &\sim \mathcal{N}_c(0, \sigma_f^2(\mathbf{z}_n)),
\label{decoder_VAE}
\end{align}
where $\mathcal{N}(\bs{\mu}, \bs{\Sigma})$ denotes the multivariate Gaussian distribution for a real-valued random vector, $\mathbf{I}$ is the identity matrix of appropriate size, and $\sigma_f^2: \mathbb{R}^L \mapsto \mathbb{R}_+$, $f \in \{0,...,F-1\}$, represents a non-linear function parametrized by $\bs{\theta}_s$. The output of this function is provided by one neuron in the $F$-dimensional output layer of a neural network, whose weights and biases are denoted by $\bs{\theta}_s$. In the context of variational autoencoders, the model \eqref{prior_VAE}-\eqref{decoder_VAE} is called the \emph{generative network} or \emph{probabilistic decoder} \cite{kingma2013auto}. 

As in the case of an NMF-based supervised variance model, we assume that the parameters $\bs{\theta}_s$ are learned on a clean speech database. This is done using a \emph{recognition network} or \emph{probabilistic encoder} in addition to the probabilistic decoder, as will be explained in Section~\ref{sec:training}. Only the latent random vectors $\{\mathbf{z}_n\}_{n=0}^{N-1}$ will have to be inferred from the observed data. Moreover, the latent dimension $L$ is set much smaller than $F$.

\subsection{Unsupervised Noise Model}

As introduced in Section~\ref{sec:varianceModelingFramework}, we use an unsupervised NMF-based Gaussian model for the noise. Independently for all $(f,n) \in \mathbb{B}$:
\begin{equation}
b_{fn} ; \mathbf{w}_{b,f}, \mathbf{h}_{b,n} \sim \mathcal{N}_c\left(0, \left(\mathbf{W}_b\mathbf{H}_b\right)_{f,n}\right).
\end{equation}

\subsection{Mixture Model} 

The observed mixture signal is modeled as follows for all $(f,n) \in \mathbb{B}$:
\begin{equation}
x_{fn} = \sqrt{g_n} s_{fn} + b_{fn},
\label{mixture}
\end{equation}
where $g_n \in \mathbb{R}_+$ represents a frame-dependent but frequency-independent gain. We introduce this term in order to provide some robustness with respect to the loudness of the training examples used to learn the variational autoencoder parameters $\bs{\theta}_s$ (this will be illustrated in Section~\ref{sec:experiments}). The speech and noise signals are further supposed to be mutually independent given the latent random vectors $\{\mathbf{z}_n\}_{n=0}^{N-1}$, such that for all $(f,n) \in \mathbb{B}$:
\begin{equation}
x_{fn} \mid \mathbf{z}_n ; \bs{\theta}_s, \bs{\theta}_{u,fn} \sim \mathcal{N}_c\left(0, g_n \sigma_f^2(\mathbf{z}_n) + \left(\mathbf{W}_b\mathbf{H}_b\right)_{f,n}\right),
\label{likelihood}
\end{equation}
where $\bs{\theta}_{u,fn} = \{\mathbf{w}_{b,f}, \mathbf{h}_{b,n}, g_n\}$ is the set of unsupervised model parameters at time-frequency point $(f,n)$.

\section{Inference}
\label{sec:inference}

Let us introduce the following notations:
\begin{itemize}[leftmargin=*]
\item $\mathbf{x}_n = \{x_{fn}\}_{f=0}^{F-1}$: The set of mixture STFT coefficients for a given time frame $n \in \{0,...,N-1\}$;
\item $\mathbf{x} = \{x_{fn}\}_{(f,n) \in \mathbb{B}}$: The set of all mixture STFT coefficients;
\item $\mathbf{z} = \{\mathbf{z}_n\}_{n=0}^{N-1}$: The set of all latent variables;
\item $\bs{\theta}_u \hspace{-0.05cm}=\hspace{-0.05cm} \{\bs{\theta}_{u,fn}\}_{(f,n) \in \mathbb{B}}$: The set of unsupervised model parameters.
\end{itemize}
From the set of observations $\mathbf{x}$, our primary goal is to estimate the unsupervised model parameters $\bs{\theta}_u$, which will serve in the final inference of the speech STFT coefficients. Unfortunately, straightforward maximum likelihood estimation is here intractable. A common alternative then consists in exploiting the latent variable structure of the model to derive an expectation-maximization (EM) algorithm \cite{dempster1977maximum}. From an initialization $\bs{\theta}_u^\star$ of the model parameters, it consists in iterating the two following steps until convergence:
\begin{itemize}[leftmargin=*]
\item E-Step: Compute $Q(\bs{\theta}_u; \bs{\theta}_u^\star) \hspace{-0.05cm}=\hspace{-0.05cm} \mathbb{E}_{p(\mathbf{z} | \mathbf{x} ; \bs{\theta}_s, \bs{\theta}_u^\star)} \hspace{-0.09cm} \left[ \ln p(\mathbf{x}, \mathbf{z} ; \bs{\theta}_s, \bs{\theta}_u) \right]$;
\item M-Step: Update $\bs{\theta}_u^\star \leftarrow \argmax_{\bs{\theta}_u} \, Q(\bs{\theta}_u; \bs{\theta}_u^\star)$.
\end{itemize}

\subsection{E-Step}
\label{subsec:EStep}

Due to the non-linear relation between the observations and the latent variables \eqref{likelihood}, we cannot compute the posterior distribution ${p(\mathbf{z} \mid \mathbf{x} ; \bs{\theta}_s, \bs{\theta}_u^\star)}$ in an analytical form. Therefore we cannot compute the expectation involved in the definition of $Q(\bs{\theta}_u; \bs{\theta}_u^\star)$ at the E-Step. We consequently approximate it by the following empirical average (Monte Carlo approximation):
\begin{align}
Q(\bs{\theta}_u; \bs{\theta}_u^\star) &\approx \tilde{Q}(\bs{\theta}_u; \bs{\theta}_u^\star) \nonumber \\
& \overset{c}{=} - \frac{1}{R} \sum_{r=1}^{R} \sum_{(f,n) \in \mathbb{B}} \bigg[ \ln\left( g_n \sigma_f^2\left(\mathbf{z}_n^{(r)}\right) + \left(\mathbf{W}_b\mathbf{H}_b\right)_{f,n} \right) \nonumber \\
& \hspace{1cm} + \frac{|x_{fn}|^2}{g_n \sigma_f^2\left(\mathbf{z}_n^{(r)}\right) + \left(\mathbf{W}_b\mathbf{H}_b\right)_{f,n}} \bigg],
\label{QFunction_approx}
\end{align}
where $\overset{c}{=}$ denotes equality up to additive constants with respect to $\bs{\theta}_u$ and $\bs{\theta}_u^\star$, and $\{\mathbf{z}_n^{(r)}\}_{r=1,...,R}$ is a sequence of samples drawn from the posterior $p(\mathbf{z}_n \mid \mathbf{x}_n ; \bs{\theta}_s, \bs{\theta}_u^\star)$ using a Markov Chain Monte Carlo (MCMC) method. Here we use the Metropolis-Hastings algorithm \cite{Robert:2005:MCS:1051451}. This approach forms the basis of the Monte Carlo EM (MCEM) algorithm \cite{wei1990monte}. Note that unlike the standard EM algorithm, it does not ensure that the likelihood increases at each iteration. Nevertheless, some convergence results in terms of stationary point of the likelihood can be obtained under suitable conditions \cite{chan1995monte}.

At the $m$-th iteration of the Metropolis-Hastings algorithm and independently for all $n \in \{0,...,N-1\}$, we first draw a sample $\mathbf{z}_n$ from a proposal random walk distribution:
\begin{equation}
\mathbf{z}_n \mid \mathbf{z}_n^{(m-1)} ; \epsilon^2 \sim \mathcal{N}(\mathbf{z}_n^{(m-1)}, \epsilon^2 \mathbf{I}).
\end{equation}
Using the fact that this is a symmetric proposal distribution \cite{Robert:2005:MCS:1051451}, we then compute the acceptance probability $\alpha$ defined by: 
\begin{align}
\alpha =&\min\left(1, \frac{p\left( \mathbf{x}_n \mid \mathbf{z}_n ; \bs{\theta}_s, \bs{\theta}_u^\star \right) p\left(\mathbf{z}_n\right)}{p\left( \mathbf{x}_n \mid \mathbf{z}_n^{(m-1)} ; \bs{\theta}_s, \bs{\theta}_u^\star) \right) p\left(\mathbf{z}_n^{(m-1)}\right)}\right),
\label{log_acceptance_ratio}
\end{align}
where $p\left( \mathbf{x}_n \mid \mathbf{z}_n ; \bs{\theta}_s, \bs{\theta}_u^\star \right) = \prod_{f=0}^{F-1} p(x_{fn} \mid \mathbf{z}_n ; \bs{\theta}_s, \bs{\theta}_{u,fn}^\star)$ with $p(x_{fn} \mid \mathbf{z}_n ; \bs{\theta}_s, \bs{\theta}_{u,fn}^\star)$ defined in \eqref{likelihood} and $p\left(\mathbf{z}_n\right)$ defined in \eqref{prior_VAE}. Then we draw $u$ from the uniform distribution $\mathcal{U}([0,1])$. If $u < \alpha$, we accept the sample and set $\mathbf{z}_n^{(m)} = \mathbf{z}_n$, otherwise we reject the sample and set $\mathbf{z}_n^{(m)} = \mathbf{z}_n^{(m-1)}$. We only keep the last $R$ samples for computing $\tilde{Q}(\bs{\theta}_u; \bs{\theta}_u^\star)$ in \eqref{QFunction_approx}, i.e. we discard the samples drawn during a so called burn-in period. 

\subsection{M-Step}

At the M-step of the MCEM algorithm, we want to maximize $\tilde{Q}(\bs{\theta}_u; \bs{\theta}_u^\star)$ in \eqref{QFunction_approx} with respect to the unsupervised model parameters $\bs{\theta}_u$. As usual in the NMF literature \cite{fevotte2011algorithms}, we adopt a block-coordinate approach by successively and individually updating $\mathbf{H}_b$, $\mathbf{W}_b$ and $\mathbf{g} = [g_0, ..., g_{N-1}]^\top$, using the auxiliary function technique. We will provide derivation details only for $\mathbf{H}_b$, as the procedure is very similar for the other parameters.

Let $\mathcal{C}(\mathbf{H}_b)$ be the opposite of $\tilde{Q}(\bs{\theta}_u; \bs{\theta}_u^\star)$ in \eqref{QFunction_approx}, seen as a function of $\mathbf{H}_b$ only (other parameters are fixed). Our goal is to minimize $\mathcal{C}(\mathbf{H}_b)$ under a non-negativity constraint.

\begin{definition}[Auxiliary function]
The $ \mathbb{R}_+^{K_b \times N} \times \mathbb{R}_+^{K_b \times N} \mapsto \mathbb{R}_+$ mapping $\mathcal{G}(\mathbf{H}_b, \tilde{\mathbf{H}}_b)$ is an auxiliary function to $\mathcal{C}(\mathbf{H}_b)$ if and only if
\begin{align}
&\forall (\mathbf{H}_b, \tilde{\mathbf{H}}_b) \in \mathbb{R}_+^{K_b \times N} \times \mathbb{R}_+^{K_b \times N}, &\mathcal{C}(\mathbf{H}_b) \le \mathcal{G}(\mathbf{H}_b, \tilde{\mathbf{H}}_b);\\
&\forall \mathbf{H}_b \in \mathbb{R}_+^{K_b \times N},  &\mathcal{C}(\mathbf{H}_b) = \mathcal{G}(\mathbf{H}_b, \mathbf{H}_b).
\end{align}
\end{definition}
\noindent In other words, $\mathcal{G}(\mathbf{H}_b, \tilde{\mathbf{H}}_b)$ is an upper bound of $\mathcal{C}(\mathbf{H}_b)$ which is tight for $\tilde{\mathbf{H}}_b = \mathbf{H}_b$. The original minimization problem can be replaced with an alternate minimization of this upper bound. From an initial point $\mathbf{H}_b^\star$ we iterate: 
\begin{equation}
\mathbf{H}_b^\star \leftarrow  \argmin\limits_{\mathbf{H}_b  \in \mathbb{R}_+^{K_b \times N}}\, \mathcal{G}(\mathbf{H}_b, \mathbf{H}_b^\star).
\label{MM_update}
\end{equation}
This procedure corresponds to the majorize-minimize (MM) algorithm \cite{hunter2004tutorial}, which by construction leads to a monotonic decrease of $\mathcal{C}(\mathbf{H}_b)$. Moreover, its convergence properties are the same as the ones of the EM algorithm \cite{wu1983convergence}.

\begin{proposition}[Auxiliary function to $\mathcal{C}(\mathbf{H}_b)$] 
The function $\mathcal{G}(\mathbf{H}_b, \tilde{\mathbf{H}}_b)$ defined below is an auxiliary function to $\mathcal{C}(\mathbf{H}_b)$.
\begin{align}
\mathcal{G}(\mathbf{H}_b, \tilde{\mathbf{H}}_b)& = \frac{1}{R} \sum_{r=1}^{R} \sum_{(f,n) \in \mathbb{B}} \Bigg[ \ln\left( g_n \sigma_f^2\left(\mathbf{z}_n^{(r)}\right) + \left(\mathbf{W}_b\tilde{\mathbf{H}}_b\right)_{f,n} \right)  \nonumber \\
& +  \frac{\left(\mathbf{W}_b\mathbf{H}_b\right)_{f,n} - \left(\mathbf{W}_b\tilde{\mathbf{H}}_b\right)_{f,n}}{g_n \sigma_f^2\left(\mathbf{z}_n^{(r)}\right) + \left(\mathbf{W}_b\tilde{\mathbf{H}}_b\right)_{f,n} } \nonumber \\
&+ |x_{fn}|^2 \left( \frac{g_n \sigma_f^2\left(\mathbf{z}_n^{(r)}\right)}{\left(g_n \sigma_f^2\left(\mathbf{z}_n^{(r)}\right) + \left(\mathbf{W}_b\tilde{\mathbf{H}}_b\right)_{f,n}\right)^2} \right. \nonumber \\
&\left. + \sum_{k=1}^{K_b} \frac{w_{b,fk} \tilde{h}_{b,kn}^2}{h_{b,kn} \left(g_n \sigma_f^2\left(\mathbf{z}_n^{(r)}\right) + \left(\mathbf{W}_b\tilde{\mathbf{H}}_b\right)_{f,n}\right)^2} \right) \Bigg],
\end{align}
where for all $(f,n) \in \mathbb{B}$ and $k \in \{1,...,K_b\}$, $w_{b,fk} = (\mathbf{W}_b)_{f,k}$ $h_{b,kn} = (\mathbf{H}_b)_{k,n}$ and $\tilde{h}_{b,kn} = (\tilde{\mathbf{H}}_b)_{k,n}$.
\end{proposition}

\begin{proof}\renewcommand{\qedsymbol}{}
The proof is provided in \cite{supMat} due to space limitation.
\end{proof}

The auxiliary function $\mathcal{G}(\mathbf{H}_b, \tilde{\mathbf{H}}_b)$ is separable in convex functions of the individual coefficients $h_{b,kn}$, $k \in \{1,...,K_b\}$, $n \in \{0,...,N-1\}$, which is convenient for updating those parameters. By canceling the partial derivative of $\mathcal{G}(\mathbf{H}_b, \tilde{\mathbf{H}}_b)$ with respect to $h_{b,kn}$, we obtain an update which depends on $\tilde{h}_{b,kn}$. According to the MM algorithm, we can use the fact that $\tilde{h}_{b,kn}$ is equal to the previous value of $h_{b,kn}$ for obtaining the following multiplicative update in matrix form:
\begin{equation}
\mathbf{H}_b \leftarrow \mathbf{H}_b \odot \left[ \frac{\mathbf{W}_b^\top \left( \mid \mathbf{X} \mid^{\odot 2} \odot \sum\limits_{r=1}^{R} \left(\mathbf{V}_x^{(r)} \right)^{\odot-2} \right)}{\mathbf{W}_b^\top \sum\limits_{r=1}^{R} \left(\mathbf{V}_x^{(r)}\right)^{\odot-1} } \right]^{\odot 1/2},
\label{updateH}
\end{equation}
where $\odot$ denotes element-wise multiplication and exponentiation, matrix division is also element-wise, $\mathbf{V}_x^{(r)} \in \mathbb{R}_+^{F \times N}$ is the matrix of entries $\left(\mathbf{V}_x^{(r)}\right)_{f,n} = g_n \sigma_f^2\left(\mathbf{z}_n^{(r)}\right) + \left(\mathbf{W}_b\mathbf{H}_b\right)_{f,n}$, and \linebreak $\mathbf{X} \in \mathbb{C}^{F \times N}$ is the matrix of entries $ \left(\mathbf{X}\right)_{f,n} = x_{fn}$. We can straightforwardly apply the same procedure for maximizing $\tilde{Q}(\bs{\theta}_u; \bs{\theta}_u^\star)$ with respect to $\mathbf{W}_b$. The resulting multiplicative update rule is given by:
\begin{equation}
\mathbf{W}_b \leftarrow \mathbf{W}_b \odot \left[ \frac{\left( \mid \mathbf{X} \mid^{\odot 2} \odot \sum\limits_{r=1}^{R} \left(\mathbf{V}_x^{(r)}\right)^{\odot-2} \right) \mathbf{H}_b^\top}{\sum\limits_{r=1}^{R} \left(\mathbf{V}_x^{(r)}\right)^{\odot-1} \mathbf{H}_b^\top } \right]^{\odot 1/2}.
\label{updateW}
\end{equation}

Again, in a very similar fashion, the vector of speech gains is updated as follows:
\begin{equation}
\mathbf{g}^\top \leftarrow \mathbf{g}^\top \odot \left[ \frac{ \mathbf{1}^\top  \left[\mid \mathbf{X} \mid^{\odot 2} \odot \sum\limits_{r=1}^{R} \left(\mathbf{V}_s^{(r)} \odot \left(\mathbf{V}_x^{(r)} \right)^{\odot-2}\right)\right]}{\mathbf{1}^\top \left[ \sum\limits_{r=1}^{R} \left(\mathbf{V}_s^{(r)} \odot \left(\mathbf{V}_x^{(r)} \right)^{\odot-1}\right)\right]} \right]^{\odot 1/2},
\label{update_g}
\end{equation}
where $\mathbf{1}$ is an all-ones column vector of dimension $F$ and $\mathbf{V}_s^{(r)} \in \mathbb{R}_+^{F \times N}$ is the matrix of entries $\left(\mathbf{V}_s^{(r)}\right)_{f,n} = \sigma_f^2\left(\mathbf{z}_n^{(r)}\right) $. The non-negativity of $\mathbf{H}_b$, $\mathbf{W}_b$ and $\mathbf{g}$ is ensured provided that those parameters are initialized with non-negative entries. In practice, at the M-step of the MCEM algorithm, we perform only one iteration of updates \eqref{updateH}, \eqref{updateW} and \eqref{update_g}.

\subsection{Speech reconstruction}
\label{subsec:speech_rec}

In the following $\bs{\theta}_u^\star = \{\mathbf{W}_b^\star, \mathbf{H}_b^\star, \mathbf{g}^\star\}$ denotes the set of parameters estimated by the above MCEM algorithm. For all $(f,n) \in \mathbb{B}$, let $\tilde{s}_{fn} = \sqrt{g_n^\star} s_{fn}$ be the scaled version of the speech STFT coefficients as introduced in \eqref{mixture}, with $g_n^\star = (\mathbf{g}^\star)_n$. Our final goal is to estimate those coefficients according to their posterior mean:
\begin{align}
\hat{\tilde{s}}_{fn} &= \mathbb{E}_{p(\tilde{s}_{fn} \mid x_{fn} ; \bs{\theta}_s, \bs{\theta}_u^\star)}  [\tilde{s}_{fn}] \nonumber \\
&= \mathbb{E}_{p(\mathbf{z}_n \mid \mathbf{x}_n ; \bs{\theta}_s, \bs{\theta}_u^\star) }\left[ \mathbb{E}_{p(\tilde{s}_{fn} \mid \mathbf{z}_n, \mathbf{x}_n ; \bs{\theta}_s, \bs{\theta}_u^\star)} [\tilde{s}_{fn}] \right]\nonumber \\
&= \mathbb{E}_{p(\mathbf{z}_n \mid \mathbf{x}_n ; \bs{\theta}_s, \bs{\theta}_u^\star) }\left[\frac{g_n^\star \sigma_f^2(\mathbf{z}_n)}{g_n^\star \sigma_f^2(\mathbf{z}_n) + (\mathbf{W}_b^\star\mathbf{H}_b^\star)_{f,n}}\right]x_{fn}.
\label{source_estimate}
\end{align}
This estimation corresponds to a soft-masking of the mixture signal, similarly as the standard Wiener filtering. Exactly as before, this expectation cannot be computed in an analytical form, but we can approximate it using the same Metropolis-Hastings algorithm as detailed in Section~\ref{subsec:EStep}. Note that this estimation procedure is different than the one proposed in \cite{bando2017statistical}. Here the STFT source coefficients are estimated from their true posterior distribution, without conditioning on the latent variables. The time-domain estimate of the speech signal is finally obtained by inverse STFT and overlap-add.

\section{Supervised Training of the Speech Model }
\label{sec:training}

For supervised training of the speech model, we use a large dataset of independent clean-speech STFT time frames $\mathbf{s}^{tr} = \{\mathbf{s}_n \in \mathbb{C}^F \}_{n=0}^{N_{tr}-1}$. In our speech enhancement method detailed above, we are only interested in using the generative model defined in \eqref{prior_VAE} and \eqref{decoder_VAE}. However, in order to learn the parameters $\bs{\theta}_s$ of this generative model, we also need to introduce a \emph{recognition network} or \emph{probabilistic encoder}  $q(\mathbf{z} \mid \mathbf{s}^{tr} ; \bs{\phi})$, which is an approximation of the true posterior $p(\mathbf{z} \mid \mathbf{s}^{tr} ; \bs{\theta}_s)$ \cite{kingma2013auto}. Independently for all $l \in \{0,...,L-1\}$ and $n \in \{0,...,N_{tr}-1\}$, $q(\mathbf{z} \mid \mathbf{s}^{tr} ; \bs{\phi})$ is defined similarly as in \cite{bando2017statistical} by:
\begin{equation}
z_{l,n} \mid \mathbf{s}_n ; \bs{\phi} \sim \mathcal{N}\left(\tilde{\mu}_l\left(|\mathbf{s}_n|^{\odot 2}\right), \tilde{\sigma}_l^2\left(|\mathbf{s}_n|^{\odot 2}\right) \right),
\label{encoder_VAE}
\end{equation}
where $z_{l,n} = (\mathbf{z}_n)_l$ and $\tilde{\mu}_l: \mathbb{R}_+^{F} \mapsto \mathbb{R}$, $\tilde{\sigma}_l^2: \mathbb{R}_+^{F} \mapsto \mathbb{R}_+$ are non-linear functions parametrized by $\bs{\phi}$. The output of each function is provided by one neuron in the $2L$-dimensional output layer of a neural network, whose weights and biases are denoted by $\bs{\phi}$.

As explained in \cite{kingma2013auto}, the recognition network parameters $\bs{\phi}$ and the generative network parameters $\bs{\theta}_s$ can be jointly trained by maximizing the evidence lower-bound (ELBO) defined by:
\begin{align}
\mathcal{L}\left(\bs{\theta}_s, \bs{\phi}\right) &=  \E_{q\left(\mathbf{z} \mid \mathbf{s}^{tr} ; \bs{\phi}\right)}\left[ \ln p\left(\mathbf{s}^{tr} \mid \mathbf{z} ; \bs{\theta}_s \right) \right] \nonumber \\
& \hspace{1.5cm} - D_{KL}\left(q\left(\mathbf{z} \mid \mathbf{s}^{tr} ; \bs{\phi}\right) \parallel p(\mathbf{z})\right),
\label{ELBO_1}
\end{align}
where $D_{KL}(q \parallel p) = \mathbb{E}_q[\ln(q/p)]$ is the Kullback-Leibler divergence. The first term in the right-hand side of \eqref{ELBO_1} is a reconstruction accuracy term, while the second one is a regularization term. From \eqref{prior_VAE}, \eqref{decoder_VAE} and \eqref{encoder_VAE}, the ELBO can be developed as follows:
\begin{align}
\mathcal{L}\left(\bs{\theta}_s, \bs{\phi}\right) \overset{c}{=}&- \sum_{f=0}^{F-1}\sum_{n=0}^{N_{tr}-1} \mathbb{E}_{q\left(\mathbf{z}_n \mid \mathbf{s}_n ; \bs{\phi} \right)}\left[ d_{IS}\left(\left|s_{fn}\right|^2 ; \sigma_f^2(\mathbf{z}_n)\right) \right] \nonumber \\
& \hspace{-1.6cm} + \frac{1}{2} \sum_{l=0}^{L-1} \sum_{n=0}^{N_{tr}-1}\left[ \ln \tilde{\sigma}_l^2\left(\left|\mathbf{s}_n\right|^{\odot 2}\right) - \tilde{\mu}_l\left(\left|\mathbf{s}_n\right|^{\odot 2}\right)^2 - \tilde{\sigma}_l^2\left(\left|\mathbf{s}_n\right|^{\odot 2}\right) \right],
\label{ELBO}
\end{align}
where $d_{IS}(x;y) = x/y - \ln(x/y) - 1$ is the Itakura-Saito (IS) divergence. We note that maximizing the ELBO with respect to the parameters $\bs{\theta}_s$ of the generative speech model amounts to minimizing a cost function that involves the IS divergence between the observed power spectrogram $|s_{fn}|^2$ and the variance model $\sigma_f^2(\mathbf{z}_n)$. It is interesting to highlight the similarity with the Gaussian generative speech model \eqref{local_gaussian_source_model}-\eqref{NMF_variance_model}, where maximum likelihood estimation of its parameters also amounts to minimizing the IS divergence between $|s_{fn}|^2$ and the NMF-based variance parametrization $(\mathbf{W}_s \mathbf{H}_s)_{f,n}$, as proved in \cite{ISNMF}.

Finally, as the expectation in \eqref{ELBO} cannot be computed in an analytical form, we again use a Monte Carlo approximation:
\medmuskip=1mu
\thinmuskip=1mu
\thickmuskip=2mu
\begin{equation}
\mathbb{E}_{q\left(\mathbf{z}_n \mid \mathbf{s}_n ; \bs{\phi} \right)}\left[ d_{IS}\left(\left|s_{fn}\right|^2 ; \sigma_f^2(\mathbf{z}_n)\right) \right] \approx \frac{1}{R} \sum_{r=1}^{R} d_{IS}\left(\left|s_{fn}\right|^2 ; \sigma_f^2\left(\mathbf{z}_n^{(r)}\right)\right),
\label{Monte_Carlo_estimate_VAE}
\end{equation}
\medmuskip=4mu plus 2mu minus 4mu
\thinmuskip=3mu
\thickmuskip=5mu plus 5mu
where for all $r \in \{1,...,R\}$, $\mathbf{z}_n^{(r)}$ is independently drawn from $q\left(\mathbf{z}_n \mid \mathbf{s}_n ; \bs{\phi} \right)$ in \eqref{encoder_VAE}, using the so-called \emph{reparametrization trick} \cite{kingma2013auto}. Independently for all $l \in \{0,...,L-1\}$, it consists in sampling:
\begin{equation}
z_{l,n}^{(r)} \mid \mathbf{s}_n ; \bs{\phi} \sim \tilde{\mu}_l\left(|\mathbf{s}_n|^{\odot 2}\right) + \tilde{\sigma}_l\left(|\mathbf{s}_n|^{\odot 2}\right) \mathcal{N}(0, 1).
\end{equation}
Injecting \eqref{Monte_Carlo_estimate_VAE} in \eqref{ELBO}, we obtain an objective function which is differentiable with respect to both $\bs{\theta}_s$ and $\bs{\phi}$, and which can be optimized using gradient-ascent-based algorithms.

\section{Experiments}
\label{sec:experiments}

\begin{figure}[t]
\centering
\includegraphics[width=.8\linewidth]{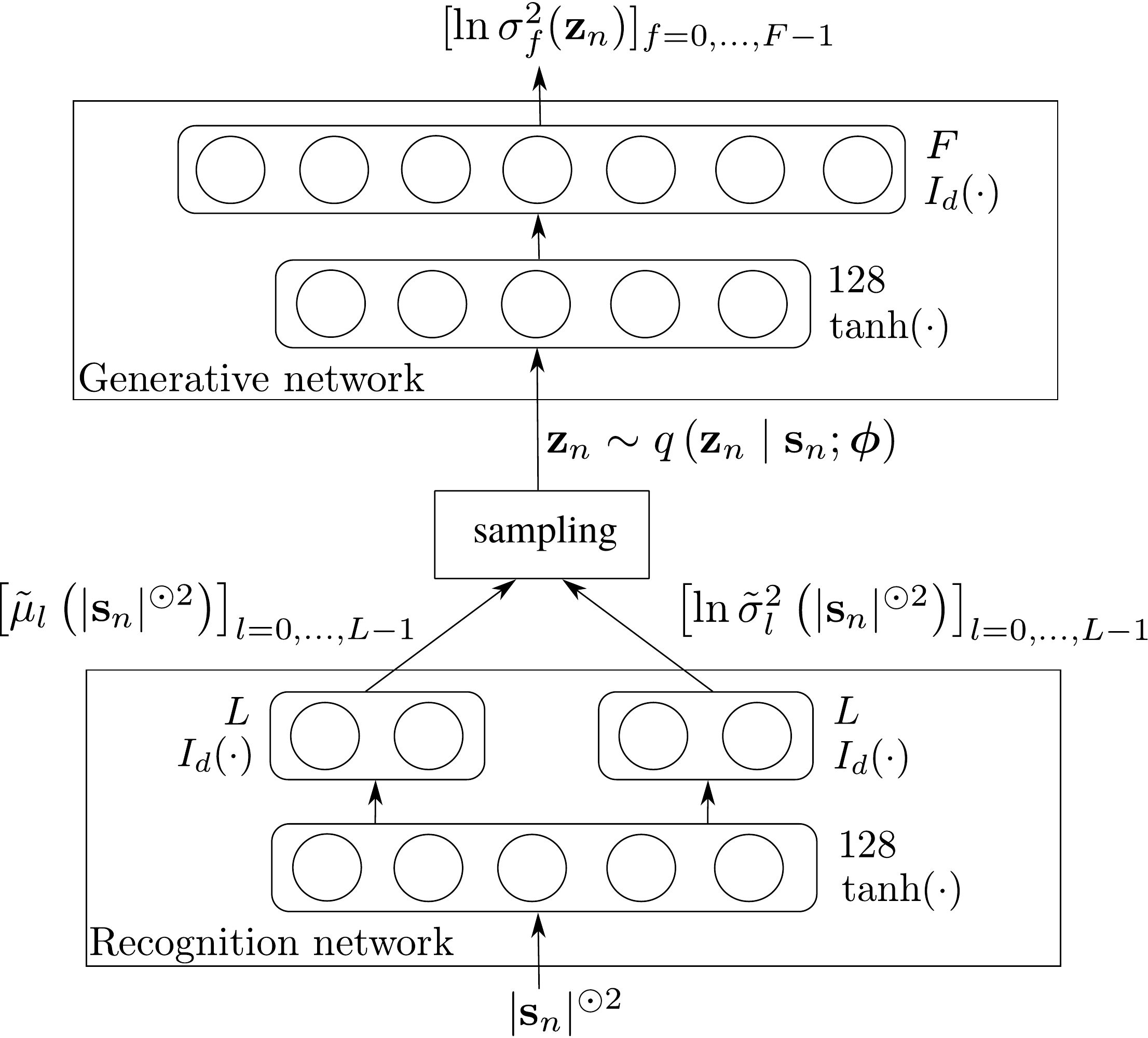}
\caption{Structure of our variational autoencoder. All layers are fully connected. Next to each layer we indicate its size along with the activation function that is used.
}%
\label{fig:network_architecture}%
\end{figure}

\textit{Database}: The supervised speech model parameters are learned from the training set of the TIMIT database \cite{TIMIT}. It contains almost 4 hours of speech signals at a 16-kHz sampling rate, distributed over 462 speakers. For the evaluation of the speech enhancement algorithm, we mixed clean speech signals from the test set of the TIMIT database and noise signals from the DEMAND database \cite{DEMAND_db}, with various noisy environments: domestic, nature, office, indoor public spaces, street and transportation. We created 168 mixtures at a 0~dB signal-to-noise ratio (one mixture per speaker in the TIMIT test set). Note that both speakers and sentences are different than in the training set.

\textit{Baseline methods}: Considering that both approaches belong to the variance modeling framework, we first compare our method with a semi-supervised NMF baseline as described in Section~\ref{sec:varianceModelingFramework}. The speech NMF dictionary is learned using the training set of the TIMIT database, and only the speech activation matrix and the noise NMF model parameters are estimated from the noisy mixture signals. The speech signal is then recovered by Wiener filtering. We also compare our approach with the fully supervised deep learning method proposed in \cite{xu2015regression}. In this work, a deep neural network is trained to map noisy speech log-power spectrograms to clean speech log-power spectrograms. The authors showed that their mapping-based approach outperformed other deep learning techniques that rely on the estimation of time-frequency masks. Moreover, the authors used more than 100 different noise types to train their system, which showed to be quite effective in handling unseen noise types. It is therefore a very relevant method to compare with. We used the Python code provided by the authors.\footnote{\url{https://github.com/yongxuUSTC/sednn}}

\textit{Parameter settings}: The STFT is computed using a 64-ms sine window (i.e.~$F=513$) with 75\%-overlap. We compare the performance of our method with the semi-supervised NMF baseline according to the dimension of the latent space $L$ and the rank $K_s$ of the NMF speech model. For fair comparison, we set $L = K_s \in \{8, 16, 32, 64, 128\}$. The rank of the noise model is arbitrarily fixed to $K_b = 10$. Unsupervised NMF parameters are randomly initialized. For the proposed method, the gain vector $\mathbf{g}$ is initialized with an all-ones vector. The iterative algorithms of those two methods are stopped if the improvement in terms of objective function is smaller that $10^{-4}$ (for our method we monitored $\tilde{Q}(\bs{\theta}_u; \bs{\theta}_u^\star)$). The parameters of our MCEM algorithm are the following ones: At each E-Step (see Section~\ref{subsec:EStep}), we run 40 iterations of the Metropolis-Hastings algorithm using $\epsilon^2 = 0.01$ for the proposal distribution, and we discard the first 30 samples as the burn-in period. At the first iteration of the MCEM algorithm we need to initialize the Markov chain of the Metropolis-Hastings algorithm. For that purpose we use the recognition network considering the mixture signal as input because the clean speech signal is not observed: for all $l \in \{0,...,L-1\}$, $z_{l,n}^{(0)} = \tilde{\mu}_l\left(|\mathbf{x}_n|^{\odot 2}\right)$. Then, at each new E-Step, we use the last sample drawn at the previous E-Step to initialize the Markov chain. For estimating the speech STFT coefficients (see Section~\ref{subsec:speech_rec}), we run the Metropolis-Hastings algorithm for 100 iterations and discard the first 75 samples. 

\textit{Variational autoencoder}: The structure of our variational autoencoder is represented in Fig.~\ref{fig:network_architecture}. As in \cite{kingma2013auto}, we use hyperbolic tangent ($\tanh(\cdot)$) activation functions except for the encoder/decoder output layers, which use identity activation functions ($I_d(\cdot)$). The values of these output layers thus lie in $\mathbb{R}$, which is the reason why we output logarithm of variances.  For learning the parameters $\bs{\theta}_s$ and $\bs{\phi}$ (see Section~\ref{sec:training}), we used the Adam optimizer \cite{kingma2014adam} with a step size of $10^{-3}$, exponential decay rates for the first and second moment estimates of $0.9$ and $0.999$ respectively, and an epsilon of $10^{-7}$ for preventing division by zero. 20\% of the TIMIT training set was kept as a validation set, and early stopping with a patience of 10 epochs was used. Weights were initialized using the uniform initializer described in \cite{glorot2010understanding}. For reproducibility, a Python implementation of our algorithm is available online.\footnote{\url{https://github.com/sleglaive/MLSP-2018}}

\textit{Results}: We evaluate the enhanced speech signal in terms of signal-to-distortion ratio (SDR) \cite{vincent2006performance} and perceptual evaluation of speech quality (PESQ) measure \cite{rix2001perceptual}. The SDR is expressed in decibels (dB) and the PESQ score lies between $-0.5$ and $4.5$. We computed the SDR using the mir\_eval Python library.\footnote{\url{https://github.com/craffel/mir_eval}} We first illustrate the interest of having a time-frame dependent gain $g_n$ as introduced in \eqref{mixture}. Inevitably, the variational autoencoder is trained using examples with a limited loudness range. If we force $g_n = 1$ for all $n\in \{0,...,N-1\}$, we can therefore expect the speech enhancement quality to depend on the speech power in the observed noisy mixture signal. This problem is illustrated in Fig.~\ref{fig:illustrate_gains}: The SDR indicated by black dots suddenly drops for a scaling factor of the mixture signal greater than 12~dB. It indicates that the power of the underlying clean speech signal in this mixture went out of the power range observed during the training. 
One solution to this problem could be to perform some kind of data augmentation as in \cite{bando2017statistical}. Here we exploited the statistical modeling framework, by introducing parameters in order to handle some inherent drawbacks of learning-based approaches. The gray diamonds in Fig.~\ref{fig:illustrate_gains} indeed show that the time-frame-dependent gains provide a solution to this robustness issue.

The speech enhancement results are presented in Fig.~\ref{fig:results}. We first observe that while the NMF baseline requires a sufficiently large rank $K_s=64$, our method obtains quite stable results for a latent dimension $L$ greater than or equal to $16$. Moreover it always performs better than the NMF baseline in terms of both SDR and PESQ measures. We also see that it outperforms the fully supervised deep learning approach \cite{xu2015regression}. It is indeed difficult for such kind of approaches to generalize to unseen noises, which strongly justifies the interest of semi-supervised approaches. Audio examples illustrating those results are available online.\footnote{\url{https://sleglaive.github.io/demo-mlsp18.html}}

We conclude this section with some information about the computational time for enhancing a 2.6 seconds-long mixture, using a central processing unit at 2.5 GHz. One iteration of the proposed MCEM algorithm (with $L=64$) takes around 2.9 seconds. For the semi-supervised NMF baseline (with $K_s=64$), one iteration of update rules takes around 4 milliseconds. The fully supervised deep-learning approach, which is non iterative, takes around 7 seconds.

\begin{figure}[t]
\centering
\includegraphics[width=.6\linewidth]{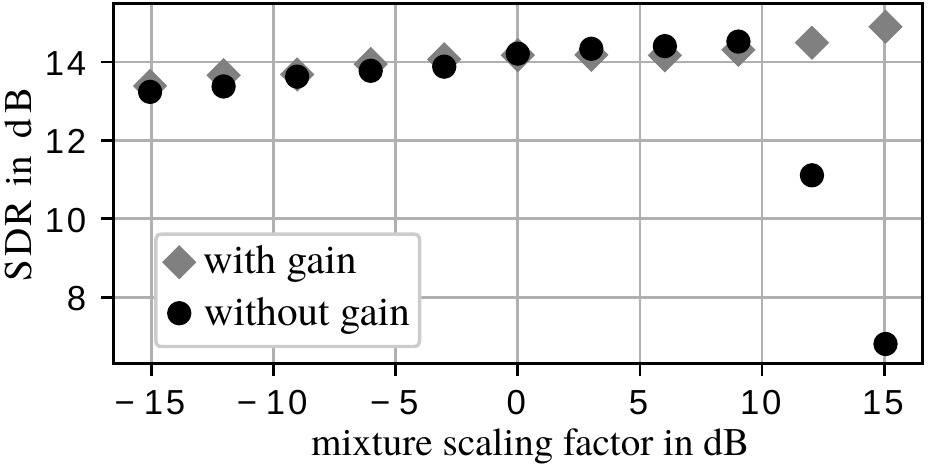}
\vspace{-.2cm}
\caption{SDR according to different scalings of the power of one mixture signal in our test dataset. Grey diamonds represent the performance when the vector of time-frame dependent gains $\mathbf{g}$ is updated at the M-Step, and black dots when $\mathbf{g}$ is fixed to an all-ones vector.}
\label{fig:illustrate_gains}
\end{figure}

\begin{figure}[t]
\centering
\includegraphics[width=.95\linewidth]{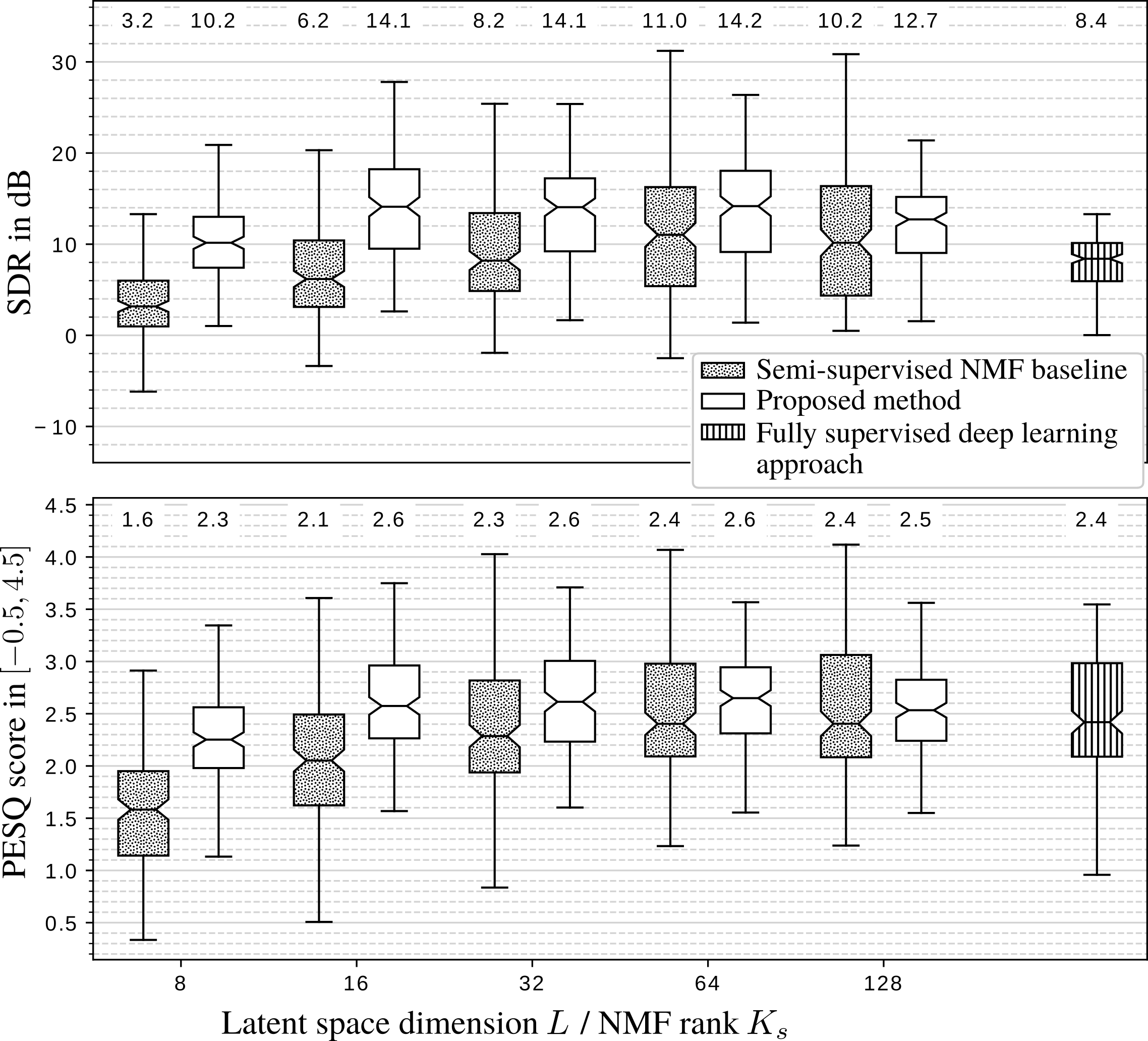}
\vspace{-.2cm}
\caption{Enhancement results, the median is indicated above each box.}
\label{fig:results}
\end{figure}

\section{Conclusion}
\label{sec:conclusion}

In this paper we showed that within the variance modeling framework, variational autoencoders are a suitable alternative to supervised NMF. We proposed a semi-supervised speech enhancement method based on an MCEM algorithm. Experimental results demonstrated that it outperforms both a semi-supervised NMF baseline and a fully supervised deep learning approach. We used a relatively shallow architecture for the variational autoencoder, we can therefore expect even better results by going deeper. 
Future works include using variational autoencoders for separating similar sources, e.g. multiple speakers. This is a much harder problem and we may have to exploit spatial information provided by multichannel mixtures to address it. Another interesting perspective would be to introduce a temporal model for the latent variables of the variational autoencoder.


\vspace{-.15cm}
\bibliographystyle{IEEEbib}
\bibliography{IEEEabrv,refs}

\end{document}